\documentclass[aps,pra,nobalancelastpage,superscriptaddress, twocolumn,10pt]{revtex4-2}
\usepackage[utf8]{inputenc}

\usepackage{color,ifthen,amsthm,amsmath,amsxtra,amsfonts,dsfont,graphicx,bm,tikz,scalerel,wasysym,bbm,amsthm,braket,physics,empheq,CircuitTikz}
\usepackage[colorlinks=true,linkcolor=blue, citecolor=blue, urlcolor=blue, bookmarks]{hyperref}
\newtheorem{theorem}{Theorem}
\theoremstyle{remark}

\def\Id{{\openone}}
\newcommand{\be}{\begin{equation}}
\newcommand{\ee}{\end{equation}}
\newcommand{\bea}{\begin{eqnarray}}
\newcommand{\eea}{\end{eqnarray}}
\newcommand{\bse}{\begin{subequations}}
\newcommand{\ese}{\end{subequations}}

\theoremstyle{plain}

\theoremstyle{plain}

\theoremstyle{plain}

\theoremstyle{plain}

\usepackage{color}
\definecolor{myred}{RGB}{232,102,102}
\definecolor{myblue}{RGB}{187,187,255}
\definecolor{myorange0}{RGB}{252,226,5}
\definecolor{myorange0c}{RGB}{255,255,255}
\definecolor{myorange}{RGB}{255,165,0}
\definecolor{mygrey}{RGB}{105,105,105}
\definecolor{OliveGreen}{RGB}{85,107,47}
\definecolor{NavyBlue}{RGB}{0,0,128}
\definecolor{mygreen}{RGB}{34,139,34}
\definecolor{myY}{RGB}{220,255,203}
\definecolor{myYO}{RGB}{255, 220, 151}
\definecolor{mygreenc}{RGB}{150,50,50}

\newcommand{\btp}{\begin{tikzpicture}[baseline=(current  bounding  box.center), scale=.7] }
\newcommand{\btps}{\begin{tikzpicture}[baseline=(current  bounding  box.center), scale=.4] }
\newcommand{\etp}{\end{tikzpicture}}

\usepackage{comment}
\usepackage[
    colorlinks=true, linkcolor=blue, citecolor=blue, urlcolor=blue, bookmarks
    ]{hyperref}
\usepackage[inline]{enumitem}
\usepackage{CircuitTikz}

\newcommand{\mcirc}{\mathbin{\scalerel*{\fullmoon}{G}}}
\newcommand{\mcircf}{\mathbin{\scalerel*{\newmoon}{G}}}

\def\Id{{\openone}}

\hypersetup{
pdftitle={Vanishing correlations in (bi)stochastic controlled circuits},
pdfsubject={Dynamics, statistical physics},
pdfauthor={Pavel Kos, Bruno Bertini and Toma\v z Prosen},
pdfkeywords={Dynamics, statistical physics, circuits, correlations}
}

\begin{document}

\title{Vanishing correlations in (bi)stochastic controlled circuits}

\author{Pavel Kos}
\affiliation{Faculty of Mathematics and Physics, University of Ljubljana, Jadranska ulica 19, 1000 Ljubljana, Slovenia}
\affiliation{Max-Planck-Institut f\"ur Quantenoptik, Hans-Kopfermann-Str. 1, 85748 Garching, Germany}
\author{Bruno Bertini}
\affiliation{School of Physics and Astronomy, University of Birmingham, Birmingham, B15 2TT, UK}
\author{Toma\v z Prosen}
\affiliation{Faculty of Mathematics and Physics, University of Ljubljana, Jadranska ulica 19, 1000 Ljubljana, Slovenia}

\begin{abstract}
We study the dynamics of circuits composed of stochastic and bistochastic controlled gates. This type of dynamics arises from quantum circuits with random controlled gates, as well as in stochastic circuits and deterministic classical cellular automata. We prove that stochastic and bistochastic controlled gates lead to two-point spatio-temporal correlation functions that vanish everywhere except when the two operators act on the same site. More generally, for multi-point correlations the two rightmost operators must act on the same site. We argue that autocorrelation, while hard to compute, typically decays exponentially towards a value that is exponentially small in the system size. Our results reveal a broad class of quantum systems that exhibit surprisingly simple correlation structures despite their complex microscopic dynamics. 
\end{abstract}
\maketitle

\section{Introduction}
In recent years, the study of interacting quantum many-body dynamics emerged as a common theme across various communities~\cite{Polkovnikov2011, eisert2015quantum, calabrese2016introduction, essler2016quench, gogolin2016equilibration, abanin2019many, bertini2021finite,  bastianello2022introduction}, with quantum circuits coming to the fore as the central theoretical playground~\cite{fisher2022random, potter2022entanglement, bertini2025exactly}. Because of the unique analytical insight that they offer, these quantum dynamical systems in discrete space-time allowed us to understand universal aspects of, for instance, quantum chaos~\cite{chan2018solution, chan2018spectral, bertini2018exact, friedman2019spectral, bertini2021random, fritzsch2021eigenstate, fritzsch2025eigenstate}, quantum (deep) thermalization~\cite{piroli2020exact, ho2022exact, ippoliti2023dynamical}, and entanglement dynamics~\cite{nahum2017quantum, bertini2019entanglement, gopalakrishnan2019unitary, piroli2020exact, zhou2022maximal}. Indeed, while fully generic circuits often lead to complex, ergodic, and intractable behavior, specific constraints can give rise to surprising solvable structures, such as dual unitarity~\cite{bertini2019exact,gopalakrishnan2019unitary} (see also the recent review~\cite{bertini2025exactly}). Another way to make analytical progress in quantum circuits  --- following a logic that is close to that of random matrix theory --- is to average over individual instances, considering various forms of random unitary circuits (RUC)~\cite{fisher2022random, potter2022entanglement}. 

One way to go beyond ergodic systems is to consider kinetically constrained models~\cite{fredrickson1984kinetic, Ritort2003Glassy, jackle1991hierarchically}, originally developed to describe the slow dynamics of classical glasses~\cite{Biroli_2013}. Their quantum versions, such as the PXP model~\cite{2018scars}, have recently attracted substantial attention as they can exhibit weak ergodicity breaking~\cite{serbyn2021quantum}. Another important model in the kinetically constrained class is the (quantum) East model~\cite{horssen2015dynamics,  roy2020strong, brighi2022hilbert,geissler2022slow,klobas2023exact, pancotti2020quantum}, which admits a quantum circuit realization in terms of controlled gates~\cite{bertini2024localised} and, in turn, contains a solvable point with accessible thermalization dynamics~\cite{bertini2024exact}. Here we take inspiration from the latter and investigate the general properties of circuits whose dynamics (classical or quantum) is determined by controlled gates. Specifically, we consider circuits of controlled gates where the operations are stochastic or bistochastic, and show how this leads to vanishing correlation functions except for the autocorrelation. Then, we demonstrate that this type of dynamics appears naturally both in random quantum circuits of controlled gates and in the context of classical dynamics.

This paper is organized as follows. In Sec.~\ref{sec:setting}, we define the system and the diagrammatic notation used. In Sec.~\ref{sec:results}, we introduce the key controlled-stochastic and controlled-bistochastic conditions and state our main results. Sec.~\ref{sec:simplifications} contains a simple graphical proof of our main result. In Sec.~\ref{sec:cases}, we present specific examples where the controlled-stochastic and controlled-bistochastic conditions apply, including random quantum controlled gates and classical cellular automata. Sec.~\ref{sec:correlations} is dedicated to the analysis of the nonzero autocorrelation function. Finally, we conclude with a discussion and outlook.

\section{Setting}
\label{sec:setting}
{ 
In this work, we consider a setting which captures quantum, classical deterministic as well as stochastic dynamics.
In the quantum regime, the system is described by a ray in a Hilbert space $\mathcal{H}$. Will consider a one-dimensional system defined on a discrete lattice hosting $2L$ dynamical variables with $q$ internal states, thus the Hilbert space of interest is $\mathcal{H}=(\mathbb{C}^q)^{\otimes{2L}}$ with dimension $q^{2L}$. The dynamics are implemented by unitary operations, which update the state of the system. 

When considering classical deterministic dynamics we only consider a subset of the states in $\mathcal{H}$. Specifically, we focus on the so called computational basis states, i.e., classical configurations of the lattice. They can be represented as strings of length $2L$ drawn from an alphabet of size $q$ (e.g., $s \in \{0, 1, \dots, q-1\}^{2L}$). Under deterministic dynamics, a given string is updated to another string. This is implemented by a permutation matrix, thus describing reversible classical computation or deterministic cellular automata.

In the context of stochastic dynamics, we describe the system as a statistical mixture of the aforementioned classical states. The system is represented by a probability vector $\ket{p}$ of length $q^{2L}$, where each non-negative element corresponds to the probability of finding the system in that computational basis state. We consider Markovian evolution, which is described by transition matrices that map probability vectors to probability vectors --- namely, stochastic matrices. To preserve probability normalization, the columns of these matrices must sum to one.

The advantage of encoding the system in a $q^{2L}$-dimensional vector evolving under the dynamics given by a matrix of size $q^{2L} \times q^{2L}$ is that it can capture all three settings above. The specific physical setting is determined by the choice of the evolution matrix --- whether unitary, permutation, or stochastic. Not only can we treat these disparate regimes within a single framework: this perspective also enables mapping one type of dynamics onto another. For example, quantum dynamics can be rigorously mapped to stochastic dynamics upon taking ensemble averages~\cite{fisher2022random}.
}

The evolution operator $\mathcal{U}$ is composed of two-site gates organized in a brickwork pattern
\be
\mathcal{U} =\mathbb{U}_e\mathbb{U}_o,\,\, \mathbb{U}_e=\bigotimes_{x=0}^{L-1} U_{x,x+1/2},\,\,\mathbb{U}_o=\bigotimes_{x=1}^{L} U_{x-1/2,x}\,.
\ee
Here $U_{x,y}$ means that the two-body matrix $U$ acts on the variables at position $x$ and $y$. In the following we take the local gate $U$ to be either unitary or stochastic depending on the choice of dynamics (quantum versus classical). {  While our results can be easily extended for arbitrary $U$ varying in space and time, we assume uniform gates throughout this manuscript for the sake of simplicity. Then the evolution up to time $t$ can be expressed as $\mathcal{U}(t)=\mathcal{U}^t$.
} Unless specifically stated, from now on we will focus on the thermodynamic limit $L\to\infty$.

To represent relevant quantities in this system, we use a convenient diagrammatic notation based on that of tensor networks, see, e.g.\ Ref.~\cite{cirac2020matrix}. We focus on two-site controlled gates, represented as a triangle with four legs
\be
U=
 \begin{tikzpicture}[baseline={([yshift=-0.6ex]current bounding box.center)},scale=0.65]
    \prop{0}{0}{colU}
  \end{tikzpicture}
  =
\sum_{i=0}^{q-1} \ketbra{i}\otimes u_i\,,
  \label{eq:gate}
  \ee
   where the single-site gate $u_i$ is applied to the target site only if the control site is in the state $\ket{i}$, with $i=0,\ldots, q-1$.

\section{Strategy and results}
\label{sec:results}
 
Let us consider the case in which the controlled single-site gates $u_i$ are stochastic, i.e. $u_i\ket{-}=\ket{-}$, where we introduced the ``flat state'' { 
\be
\ket{-}=\frac{1}{\sqrt{q}}\sum_{i=0}^{q-1} \ket{i}=
\begin{tikzpicture}[baseline={([yshift=-0.6ex]current bounding box.center)},scale=0.65]
    \gridLine{0.5}{0}{0.5}{0.5}
    \MEh{0.5}{0}
  \end{tikzpicture}
,
\ee
where the bar can point in any direction for convenience of the graphics and still depict the same (flat) state.
Note that since $\ket{- \ldots -}$ is an equal superposition of all states, it corresponds to the flat distribution in a stochastic setting (the classical maximum entropy state). Therefore, it is an invariant state under stochastic dynamics, $\mathcal{U}\ket{- \ldots -}=\ket{- \ldots -}$. In the quantum setting, this state is connected to the vectorized identity matrix, which represents the maximally mixed (maximum entropy) state (see Eq.\eqref{eq:circ} and the  discussion following it). 
It is also useful to denote by $\ket{\mcircf}$ an arbitrary normalized state from the $q-1$ dimensional subspace orthogonal to $\ket{-}$.

We are now in a position to introduce our basic relations: for stochastic $u_i$, the following \emph{controlled-stochastic} condition is satisfied{ 
\be
\label{eq:CS}
U\ket{\psi-} = \ket{\psi-} \forall \psi
\iff
 \begin{tikzpicture}[baseline={([yshift=-0.6ex]current bounding box.center)},scale=0.65]
    \prop{0}{0}{colU}
    \MEld{0.5}{-0.5}
  \end{tikzpicture} = 
   \begin{tikzpicture}[baseline={([yshift=-0.6ex]current bounding box.center)},scale=0.65]
     \gridLine{0}{-0.5}{0}{0.5}
     \gridLine{0.5}{0}{0.5}{0.5}
         \MEh{0.5}{0}
  \end{tikzpicture}\,. 
  \ee
  }
When the controlled gates $u_i$ are also bistochastic ($\bra{-}u_i=\bra{-}$), we additionally have the \emph{controlled-bistochastic} condition
{ 
  \be
  \label{eq:BCS}
  \bra{\psi-} U =\bra{\psi-}  \forall \psi
 \iff
   \begin{tikzpicture}[baseline={([yshift=-0.6ex]current bounding box.center)},scale=0.65]
    \prop{0}{0}{colU}
    \MErd{0.5}{0.5}
  \end{tikzpicture} = 
   \begin{tikzpicture}[baseline={([yshift=-0.6ex]current bounding box.center)},scale=0.65]
     \gridLine{0}{-0.5}{0}{0.5}
     \gridLine{0.5}{0}{0.5}{-0.5}
      \MEh{0.5}{0}
         \MEh{0.5}{0}
  \end{tikzpicture}. 
\ee
}
These conditions arise in many different physical settings, for instance, for \emph{deterministic} controlled cellular automata, but also --- as we show in Sec.~\ref{sec:Tostochastic} --- for one-replica averages in circuits of \emph{random} quantum {controlled gates}. In fact, Eqs.~\eqref{eq:CS} and \eqref{eq:BCS} also hold for gates that are more general than those in Eq.~\eqref{eq:gate}: in App.~\ref{app:general_gates}, we prove that the gate satisfies Eqs.~\eqref{eq:CS} and \eqref{eq:BCS} if and only if it can be written as 
\be
U = \sum_{\alpha=1}^{r} c_\alpha \otimes u_\alpha, \qquad r \le q^2+1,
\label{eq:general_form}
\ee
with $u_i$ bistochastic and $\sum_\alpha c_\alpha = \Id$. Moreover, as we discuss in App.~\ref{sec:appgencond}, Eqs.~\eqref{eq:CS} and \eqref{eq:BCS} can be generalized by substituting the identity on the r.h.s. Eq.~\eqref{eq:CS} and Eq.~\eqref{eq:BCS} by a general bistochastic matrix. 
This leads to the same type of simplifications as for the case discussed here.

Our main result is to show that Eqs.~\eqref{eq:CS} and \eqref{eq:BCS} provide strong constraints on the structure of dynamical correlations. Namely we have:
\begin{theorem}[Vanishing of two-point correlation functions]
\label{thm1}
   Consider a brickwork circuit of two-site gates $U$ and let  
    \be
    \label{eq:correlations}
    \!\!\!\! C(x,t)\!\!=\!\! \bra{- \dots - \mcircf_x - \dots -}\!\mathcal{U}^t\!\ket{- \dots - \mcircf_0 - \dots -}\!,
    \ee
    be the infinite-temperature two-point correlation functions of diagonal traceless observables at $(0,0)$ and $(x,t)$ (subscripts correspond to the operators' position). Then the following facts hold
    \begin{itemize}
        \item[(i)] Eq.~\eqref{eq:CS} $\implies$ $C(x,t)=0$ for $x>0$;
        \item[(ii)] Eq.~\eqref{eq:BCS} $\implies$ $C(x,t)=0$ for $x<0$;
        \item[(iii)] Eq.~\eqref{eq:CS} and Eq.~\eqref{eq:BCS} $\implies$ $C(x,t)=0$ for $x\neq 0$. 
    \end{itemize}
\end{theorem}
As we discuss in Sec.~\ref{sec:correlations}, the autocorrelation function $C(0,t)$ is in general hard to compute, and can exhibit different behaviors: it generically decays exponentially to something exponentially small in the system size, but it can also vanish immediately or remain constant depending on the specific choice of $\{u_i\}$.\\

More generally, we can make a similar statement for $n$-point correlation functions { 
\begin{align}
C(x_1,t_1,\ldots,x_n,&t_n)= \bra{ - \ldots - }O_{x_n}
\mathcal{U}(t_n-t_{n-1}) \dots \notag\\
&O_{x_2}\mathcal{U}(t_2-t_1)O_{x_1}\mathcal{U}(t_1)\ket{- \ldots -},
\end{align}
where $O_i$ is the application of diagonal traceless observable at position $i$ and $\mathcal{U}(t_i-t_{i-1})=\mathcal{U}^{t_i-t_{i-1}}$ implements the evolution from $t_{i-1}$ to $t_i$.
Note that since $\ket{- \ldots -}$ is an invariant state, the above equation reduces to Eq.~\eqref{eq:correlations} for $n=2$.
}

\begin{theorem}[Vanishing-multi-point correlation functions]
\label{thm2}
If both Eq.~\eqref{eq:CS} and Eq.~\eqref{eq:BCS} hold, the $n$-point space-time correlation function $C(x_1,t_1,\ldots,x_n,t_n)$ vanishes whenever the set $\{x_1,\ldots, x_n\}$ has a unique maximal value. 
\end{theorem}

Note that if multi-point correlation functions factorize due to locality, the above theorem applies to each of the clusters. We prove the theorems in the next section. 

\section{Simplifications of correlation functions}
\label{sec:simplifications}

To prove Thms.~\ref{thm1} and \ref{thm2}, let us begin by considering the two-point correlation functions of the stochastic model. The correlations are defined as in Eq.~\eqref{eq:correlations}. Graphically they can be represented as 
\be 
C(x,t)=
\begin{tikzpicture}[baseline={([yshift=-0.6ex]current bounding box.center)},scale=0.5]
      \foreach \y in {0,...,3}{
      \foreach \x in {0,...,6}{
      \prop{\x-\y}{\x+\y}{colU}}
      }
      \foreach \x in {0,...,6}{
      \MEld{\x+.5}{\x-0.5}
       \MEld{\x+.5-4}{\x-0.5+4}
      }
      \foreach \y in {0,...,2}{
      \MErd{\y-.5-3}{-\y-0.5+3}
      \MErd{\y-.5+5}{-\y-0.5+9}
      }
      \GS{-0.5}{-0.5}
       \GS{-0.5+4}{-0.5+10},
\end{tikzpicture},
\label{eq:dynamicalcorrelationave}
\ee
where we took $x=2$, $t=5$, and have already simplified the light cones using the bistochasticity of the gates ($U\ket{--}=\ket{--}, \bra{--}U=\bra{--}$).

We prove Thm.~\ref{thm1} by further simplifying this type of diagram. We use Eq.~\eqref{eq:CS}, starting from the rightmost gate, and simplify until
\be 
C(x,t)=
\begin{tikzpicture}[baseline={([yshift=-0.6ex]current bounding box.center)},scale=0.5]
      \foreach \y in {0,...,3}{
      \foreach \x in {0,...,6}{
      \pgfmathtruncatemacro{\test}{\x - \y < 3}
        \ifnum\test=1
            \prop{\x-\y}{\x+\y}{colU}
        \fi
      }}
      \prop{3}{9}{colU}
      \MEld{3+.5}{9-0.5}
      \foreach \x in {0,...,6}{
       \MEld{\x+.5-4}{\x-0.5+4}
       }
       \foreach \x in {1,...,3}{
       \MEld{2+.5}{2*\x-0.5+2}
       \MErd{2+.5}{2*\x-1-0.5+2}
      }
      \foreach \x in {0,...,2}{
      \MEld{\x+.5}{\x-0.5}
      }
      \foreach \y in {0,...,2}{
      \MErd{\y-.5-3}{-\y-0.5+3}
      }
      \GS{-0.5}{-0.5}
       \GS{-0.5+4}{-0.5+10},
\end{tikzpicture}.
\label{eq:dynamicalcorrelationave2}
\ee
Then the topmost gate simplifies in the same way, resulting in $\bra{-}\ket{\mcircf}=0$ so the correlation vanishes. 

The same happens if the topmost operator is on the left of the initial operator, this time by using Eq.~\eqref{eq:BCS}, which implies that the bottom operator contracts with $\ket{-}$, again resulting in vanishing correlations. The only nonzero two-point correlations are thus the autocorrelation functions.

The same argument can be extended to multi-point correlation functions. In particular if both Eq.~\eqref{eq:CS} and Eq.~\eqref{eq:BCS} hold,
all states $\ket{-}$ to the right from the second rightmost operator at the position $x_{n-1}$ can propagate upwards and downwards. If they meet the rightmost operator, the correlations vanish. Therefore, for the correlations not to vanish the two rightmost operators must act on the same site. Namely, $\{x_1,\ldots, x_n\}$ must have two maxima, i.e. the highest value in the set appears exactly twice.

\section{Examples}
\label{sec:cases}

Let us discuss two important families of gates giving rise to the conditions in Eqs.~\eqref{eq:CS} and~\eqref{eq:BCS}.

\subsection{Random quantum controlled gates}
\label{sec:Tostochastic}

Consider quantum random controlled gates for qubits ($q=2$) written as
\be
W = \ketbra{0}\otimes \Id + \ketbra{1} \otimes w,
\ee
where $w$ are random unitary matrices acting on a single qubit, sampled independently at each space-time point from the Haar measure. We also introduce the ``folded" picture~\cite{cirac2020matrix}, where the bra and ket of the state are arranged on top of each other. Then the evolution superoperator consists of the gates  $W \otimes W^*$. We represent by 
\be
\ket{\mcirc}=\frac{1}{\sqrt{2}}\sum_{i=0}^{1} \ket{i}\otimes\ket{i},
\label{eq:circ}
\ee
the maximally entangled state between the two superimposed copies, which corresponds to the identity operator under Choi's isomorphism. 

Under these conditions, the averaged folded gate is given by 
\be
\begin{aligned}
\!\!\!\!\mathbb{E}[W \otimes W^*] =& P_{00} \otimes (\Id \otimes \Id) + P_{01} \otimes \mathbb{E}[\Id \otimes w^*] \\
&+ P_{10} \otimes \mathbb{E}[w \otimes \Id] + P_{11} \otimes \mathbb{E}[ w \otimes w^*],
\end{aligned}
\ee
where we introduced the projectors $P_{ij} = \ketbra{i}\otimes\ketbra{j}$ acting on the two copies of the control space. 

Since the Haar average of $w$ is zero, the cross-terms vanish and we are left with 
\be
\mathbb{E}[W \otimes W^*] = P_{00} \otimes (\Id \otimes \Id) + P_{11} \otimes \mathbb{E}[w \otimes w^*].
\ee
The Haar average of $\mathbb{E}[w \otimes w^*]$ is $1/2\sum_{i,j=0}^1 \ket{ii}\bra{jj}=\ketbra{\mcirc}$~\cite{Nielsen2010}, therefore
  \be
  \mathbb{E}[W \otimes W^*] = P_{00} \otimes \Id_4+ P_{11} \otimes \ketbra{\mcirc}.
  \ee
Now note that the projectors $P_{ii}$ project to the same bra and ket state and, due to the brickwork structure, they are applied to each leg of the system. This means that we can reduce the folded space from four to two dimensional Hilbert space spanned by $\ket{0}\equiv\ket{00}$ and $\ket{1}=\ket{11}$ (the states $\ket{01},\ket{10}$ are projected out). In the smaller space, $\ketbra{\mcirc}$ is projected to $\ketbra{-}$.  This reduces the circuit to an effective classical stochastic one characterized by the gate
\begin{align}
   U = \begin{tikzpicture}[baseline={([yshift=-0.6ex]current bounding box.center)},scale=0.65]
    \prop{0}{0}{colU}
  \end{tikzpicture} 
= P_0 \otimes \Id_2 + P_1 \otimes \ketbra{-},
  \label{eq:projectedavegate2}
\end{align}
or in terms of the Eq.~\eqref{eq:gate}, by the bistochastic matrices
\be
u_0=\Id,\qquad  u_1=\frac{1}{2}
\begin{pmatrix}
    1 & 1 \\
    1 & 1
  \end{pmatrix}.
\ee
This gate satisfies controlled-bistochastic conditions in Eqs.~\eqref{eq:CS} and \eqref{eq:BCS}.

The resulting gate has an additional interesting property, namely it is left controlled bistochastic in Hadamard rotated basis. 
Namely, denoting Hadamard matrix by
\be
H = 
\frac{1}{\sqrt{2}}\begin{pmatrix}
    1 & -1 \\
    1 & 1
  \end{pmatrix},
\ee
the gate fulfills 
\be
(H\otimes H)\cdot  U\cdot (H\otimes H)= \text{SWAP}\;   U \;\text{SWAP}, 
\ee
Graphically, the conditions are
\be
 \begin{tikzpicture}[baseline={([yshift=0.4ex]current bounding box.center)},scale=0.65]
    \prop{0}{0}{colU}
    \fMEZero{-0.5}{-0.5}
  \end{tikzpicture} = 
   \begin{tikzpicture}[baseline={([yshift=0.4ex]current bounding box.center)},scale=0.65]
     \gridLine{0}{-0.5}{0}{0.5}
     \gridLine{-0.5}{0}{-0.5}{0.5}
         \fMEZero{-0.5}{0}
  \end{tikzpicture}\,, 
  \qquad 
   \begin{tikzpicture}[baseline={([yshift=0.0ex]current bounding box.center)},scale=0.65]
    \prop{0}{0}{colU}
    \fMEZero{-0.5}{0.5}
  \end{tikzpicture} = 
   \begin{tikzpicture}[baseline={([yshift=0.0ex]current bounding box.center)},scale=0.65]
     \gridLine{0}{-0.5}{0}{0.5}
     \gridLine{-0.5}{0}{-0.5}{-0.5}
      \fMEZero{-0.5}{0}   
  \end{tikzpicture}\,, 
  \label{eq:localrelations2}
\ee
which are the mirror image of Eqs.~\eqref{eq:CS} and \eqref{eq:BCS} with the difference that they involve the state $\ket{0}$ rather than $\ket{-}$.

Now we show how the quantum correlation functions reduce to those in Eq.~\eqref{eq:dynamicalcorrelationave}.
We consider dynamical correlation functions of local operators (e.g., Pauli matrices) on the infinite temperature state
\be
  C_{\mu\nu}(x,t)=\frac{1}{\tr \Id}
  \tr[\sigma^{(\mu)}(x,t)\sigma^{(\nu)}(0,0)],
\ee
where $\sigma^{(\mu)}(x,t)=\mathcal{U}(t)^\dagger \sigma^{(\mu)} \mathcal{U}(t)
$ is the Heisenberg-evolved operator. Diagrammatically, this quantity is represented by a large tensor network formed by two layers of the circuit, one with $U$ and one with $U^*$, contracted at the top and bottom with the operators $\mcirc$.

As discussed before,  averaging over the random controlled gates forces the bra and ket spaces to match, thus projecting to a sub-space spanned by $\ket{00},\ket{11}$. This process effectively turns the operator evolution onto a simpler, classical stochastic process. The Pauli operators $\sigma^{(x,y)}$ are projected to zero, while correlations between $\sigma^{(z)}$ can be non-zero. The final averaged correlation function simplifies to Eq.~\eqref{eq:dynamicalcorrelationave} and, therefore, is described by Thms.~\ref{thm1} and \ref{thm2}. 

One can generalize this approach by averaging $w$ in terms of a more restricted measure. This produces a $u_1$ with a richer structure. For example, in App.~\ref{sec:GenTiltedEast} we show how a particular choice of average leads to the  tilted stochastic East model~\cite{defazio2024exact} with a parameter $p$ and a positive tilting parameter $s\geq 0$
\be
u_1=\begin{pmatrix}
1-p&p \ e^{-s}\\p \ e^{-s}&1-p 
\end{pmatrix}.
\ee
The fact that $s\geq 0$ implies that with our averaging procedure we can only probe the inactive phase of the tilted stochastic East model.

\subsection{Controlled bistochastic gates}
We can also consider directly classical stochastic evolution with
\be
U = \sum_{i=0}^{q-1} \ketbra{i}\otimes u_i, 
\ee
and where $u_i$ are bistochastic. Such gates satisfy Eqs.~\eqref{eq:CS} and ~\eqref{eq:BCS}.

An important subset of the above class are controlled deterministic cellular automata. In the case of bits $q=2$ the only possibilities are $u_i=\Id,X$, where $X$ is the bit flip, leading to identity and CNOT gates. 
For $q=3$, these controlled permutation matrices are subsets of deterministic automata studied in Ref.~\cite{sharipov2025ergodicbehaviorsreversible3state}. 


\section{Autocorrelation functions in (bi)stochastic controlled circuits}
\label{sec:correlations}

We have seen how the conditions in Eqs.~\eqref{eq:CS} and~\eqref{eq:BCS}, imply that correlations vanish almost everywhere. Nevertheless, they do not inform about the autocorrelation ($x=0$). Here we argue that the latter is generically hard to compute and show numerically that it typically decays exponentially in time but may show different dynamical behaviors.

The diagram from Eq.~\eqref{eq:dynamicalcorrelationave} is simplified using Eq.~\eqref{eq:CS} on the right side from the operator, leading to a triangular diagram of the form 
\be 
C(0,t) =  \begin{tikzpicture}[baseline={([yshift=-0.6ex]current bounding box.center)},scale=0.5]
      \prop{0}{0}{colU}
      \prop{-1}{1}{colU}
      \prop{-2}{2}{colU}
      \prop{0}{2}{colU}
      \prop{-3}{3}{colU}
      \prop{-1}{3}{colU}
      \prop{-2}{4}{colU}
      \prop{0}{4}{colU}
       \prop{-1}{5}{colU}
       \prop{0}{6}{colU}
      \foreach \x in {0,...,3}{
       \MEld{\x+.5-4}{\x-0.5+4}
      \MErd{\x-.5-3}{-\x-0.5+3}
      }
      \foreach \x in {0,2,...,4}{
      \bendR{0.5}{\x+0.5}{\x+1.92}
      }
      \GS{0.5}{-0.5}
       \GS{0.5}{-0.5+7},
\end{tikzpicture}=\begin{tikzpicture}[baseline={([yshift=-0.6ex]current bounding box.center)},scale=0.5]
      \prop{0}{0}{colU}
      \prop{-1}{1}{colU}
      \prop{-2}{2}{colU}
      \prop{0}{2}{colU}
      \prop{-3}{3}{colU}
      \prop{-1}{3}{colU}
      \prop{-2}{4}{colU}
      \prop{0}{4}{colU}
       \prop{-1}{5}{colU}
       \prop{0}{6}{colU}
      \foreach \x in {0,...,3}{
       \MEld{\x+.5-4}{\x-0.5+4}
      \MErd{\x-.5-3}{-\x-0.5+3}
      }
       \foreach \x in {0,...,7}{
      \GS{0.5}{-0.5+\x}},
\end{tikzpicture},
\label{eq:dynamicalcorrelationavesimp}
\ee
where in the second step we restricted ourselves to $q=2$ and inserted a resolution of the identity $\Id= \ketbra{-}+\ketbra{\mcircf}$ 
in each of the right loops. Then we used Eq.~\eqref{eq:CS} and removed the vanishing contributions from $\ketbra{-}$ terms. Note that this triangular diagram has $t$ gates on each one of the equal edges ($t=4$ in the diagram shown in Eq.~\eqref{eq:dynamicalcorrelationavesimp}).

In general the diagram in Eq.~\eqref{eq:dynamicalcorrelationavesimp} is hard to simplify. We attempted to write a recursion relation using triangle diagrams where on right edge there can be either pairs of $-$ or $\mcircf$. For gates emerging from random controlled gates, however, all exponentially many such diagrams contribute. We also did not find particular simplifications emerging in the classical deterministic case.

\begin{figure}[h!]
    \centering
    \includegraphics[width=0.95\linewidth]{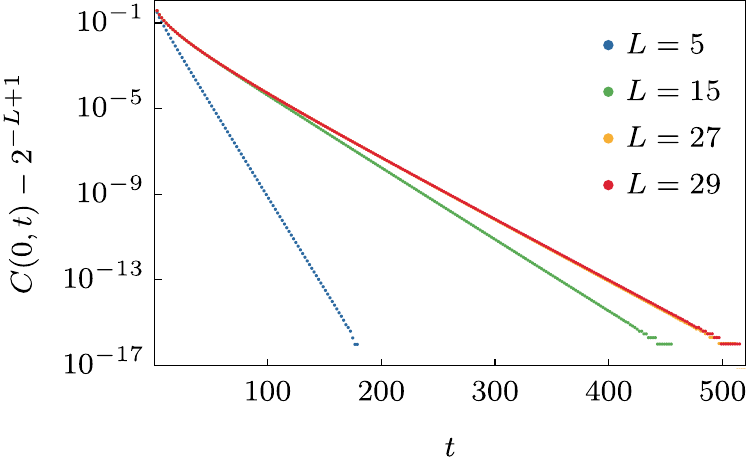}
    \caption{Exponential decay of the autocorrelation function for different system sizes and averaged controlled gate. The difference between $L=27$ and $L=29$ is almost indistinguishable. We deduced finite size contribution arising from a scar state $|00 \ldots  0 \protect\mathbin{\protect\scalerel*{\newmoon}{G}}
    \rangle$.}
    \label{fig:1}
\end{figure}

The behavior of the autocorrelation is system-dependent. While it decays exponentially in generic cases (see Fig.~\ref{fig:1}), it can decay to (system-size independent) constant values in some systems. Note that in finite systems of $L$ sites with open boundary conditions and $u_0=\Id$, a state $\ket{00 \ldots  0 \mcircf}$ is a scar eigenstate. Since $\bra{--\ldots - \mcircf}\ket{00 \ldots  0 \mcircf} = q^{-L+1}$, the autocorrelations at long time decay to this value. 

For deterministic automata with $q=2$ the non-trivial example is the CNOT gate, which corresponds to deterministic East model~\cite{klobas2023exact}. The latter is a member of the second level of hierarchical generalizations of dual unitarity~\cite{Yu_2024}, which means that its auto-correlation functions can be computed exactly giving $C(0,t)=0$ for any $t>0$. Instead, for deterministic automata with $q=3$ we have no available analytical result and we resorted to a numerical survey. The vast majority of systems display exponentially decaying autocorrelation functions. Exceptions include examples where the function is identically zero for times $t \geq 1$, and cases where the autocorrelation stabilizes at a non-zero constant. The latter examples are, however, ``trivial" because their dynamics do not update certain local states. Therefore, they feature non-decaying components in the autocorrelation.

\section{Discussion}
In this paper, we demonstrated how correlation functions of controlled stochastic models vanish away from the equal-space line. Nevertheless, the non-vanishing autocorrelation function remains generically hard to compute and displays a variety of complex behaviors. This phenomenology is a direct consequence of what we termed controlled (bi-)stochastic conditions, arising naturally both in the classical context and in the quantum one, as one-replica averages over random controlled unitaries. 

The controlled (bi-)stochastic conditions we identified are also relevant for computing local expectation values after a quantum quench, but only for special initial states featuring an ``inactive'' region of $\ket{-}$ extending to the right boundary of the system. This can be achieved, for instance, considering special bipartitioning protocols or local quenches. 

A natural extension of our work is to generalize the controlled (bi-)stochastic conditions to the multi-replica case. Such an extension would enable the analytical calculation of important quantities in many-body dynamics, including entanglement entropy and out-of-time-ordered correlators (OTOCs).

\section{Acknowledgments}
We thank Sirui Lu for help and useful discussions. This project has received funding from the European Union's Horizon Europe research and innovation programme under the Marie Skłodowska-Curie grant agreement No. 101200860 (FAQ-QuantuM2D) (P.\ K.) and by the Royal Society through the University Research Fellowship No. 201101 (B.\ B.).
 T.\ P.\ acknowledges support by European Research Council (ERC) through Advanced grant QUEST (Grant Agreement No. 101096208), and Slovenian Research and Innovation agency (ARIS) through the Program P1-0402 and Grant N1-0368.

\bibliography{bibliography}

\appendix

\setcounter{equation}{0}
\setcounter{figure}{0}
\setcounter{table}{0}
\makeatletter

\onecolumngrid
\section*{Appendix}
\appendix

\section{Characterization of the family fulfilling Eqs.~\eqref{eq:CS} and \eqref{eq:BCS}}
\label{app:general_gates}

In this appendix, we use operator-Schmidt decomposition to prove that gates satisfying the controlled-stochastic and controlled-bistochastic conditions must have the form given in Eqs.~\eqref{eq:CS} and \eqref{eq:BCS} of the main text. In fact, we consider a slight generalization of these conditions. We begin with: 

 \begin{theorem}
    \label{thm:characterization_full}
    Let $U$ be a two-site gate acting on $\mathcal{H} \otimes \mathcal{H}$ where $\mathcal{H} \simeq \mathbb C^d$. Then $U$ satisfies the controlled-stochastic condition
    \begin{align}
        U(\mathbbm{1} \otimes \ket{-}) = c \otimes \ket{-},
        \label{eq:CS_app_def}
    \end{align}
    for some single-site operator $c$ on site 1 if and only if $U$ can be written as
    \begin{align}
        U = \sum_{\alpha=1}^{r} c_\alpha \otimes u_\alpha, \qquad r \le q^2+1,
        \label{eq:general_form_app}
    \end{align}
    where $c_\alpha$ are operators on site 1 satisfying $\sum_\alpha c_\alpha = c$, and $u_\alpha$ are operators on site 2 satisfying $u_\alpha\ket{-} = \ket{-}$.
\end{theorem}

\begin{proof}
    We first prove the ``if'' direction (sufficiency), then the ``only if'' direction (necessity).\\

    \noindent \underline{Sufficiency.} Suppose $U = \sum_\alpha c_\alpha \otimes u_\alpha$ with $u_\alpha\ket{-} = \ket{-}$ and $\sum_\alpha c_\alpha = c$. Then we have
    \begin{align}
        U(\mathbbm{1} \otimes \ket{-})
         &= \sum_\alpha \left(c_\alpha \otimes u_\alpha\right) \left(\mathbbm{1} \otimes \ket{-}\right) 
         = \sum_\alpha (c_\alpha) \otimes \ket{-} 
         = \left(\sum_\alpha c_\alpha\right) \otimes \ket{-} 
         = (c) \otimes \ket{-}.
         \label{eq:sufficiency}
    \end{align}
    
     \noindent \underline{Necessity.} We begin by writing the operator-Schmidt decomposition
    \begin{align}
        U = \sum_{\alpha=1}^{r'} \lambda_\alpha A_\alpha \otimes B_\alpha,
    \end{align}
    where $r' \le q^2$, $\lambda_\alpha > 0$, and the operator families $\{A_\alpha\}$, $\{B_\alpha\}$ are linearly independent sets of operators on sites 1 and 2 respectively. Applying now Eq.~\eqref{eq:CS_app_def} we then have
    \begin{align}
        U(\mathbbm{1} \otimes \ket{-})
         &= \sum_{\alpha=1}^{r'} \lambda_\alpha (A_\alpha) \otimes \left(B_\alpha \ket{-}\right) = (c) \otimes \ket{-}.
        \label{eq:condition_expanded_app}
    \end{align}
Since the operators $\{A_\alpha\}$ are linearly independent, we can expand $c$ in terms of them to obtain 
\be
c = \sum_{\alpha=1}^{r'} \beta_\alpha A_\alpha,
\label{eq:cexp}
\ee
where we used that any operator linearly independent from $\{A_\alpha\}$ cannot appear in the expansion otherwise Eq.~\eqref{eq:condition_expanded_app} would be violated. Imposing Eq.~\eqref{eq:condition_expanded_app} we then have 
\begin{equation}
    B_\alpha |-\rangle = \mu_\alpha |-\rangle, \quad \text{for scalars } \mu_\alpha \in \mathbb{C},
\end{equation}
where we set $\mu_\alpha = \beta_\alpha/\lambda_\alpha$. Let us now consider the bipartition 
\be
\{1,\ldots,r'\}= S \cup S',
\ee
where $S$ corresponds to the values of $\alpha$ such that $\mu_\alpha\neq 0$ while $S'$ to those giving $\mu_\alpha = \beta_\alpha= 0$. Therefore, for $\alpha\in S$ we can define the stochastic matrix $u_\alpha$ as
\be
u_\alpha= \frac{1}{\mu_{\alpha}} B_\alpha, 
\ee
while for $\alpha \in S'$ we set 
\be
 u'_\alpha= B_\alpha. 
\ee
These matrices annihilate the flat state when acting upon it from the left. Writing
\be
\label{eq:uuprime}
u_\alpha = u_\alpha' + \Id, \qquad \alpha\in S', 
\ee
we then have that $u_\alpha$ are stochastic. Therefore we can write
\begin{align}
U &= \sum_{\alpha\in S} \lambda_\alpha \mu_\alpha A_\alpha \otimes u_\alpha + \sum_{\alpha\in S'} \lambda_\alpha A_\alpha \otimes u_\alpha' \label{eq:Usto_firstline} \\
& = \sum_{\alpha=1}^{r'} c_\alpha \otimes u_\alpha  - \sum_{\alpha\in S'} \lambda_\alpha A_\alpha \otimes \Id \\
& = \sum_{\alpha=1}^{r} c_\alpha \otimes u_\alpha\,.
\end{align}
where in the first step we set 
\be
c_\alpha = \begin{cases}
\lambda_\alpha \mu_\alpha A_\alpha & \alpha\in S\\
\lambda_\alpha A_\alpha & \alpha\in S'\\
\end{cases}\, ,
\ee
and in the second one 
\be
c_r =  - \sum_{\alpha\in S'} \lambda_\alpha A_\alpha, \qquad  u_r = \Id\,, \qquad r=r'+1\,. 
\ee
Recalling Eq.~\eqref{eq:cexp} we finally have 
\be
\sum_{\alpha=1}^r  c_\alpha = \sum_{\alpha=1}^{r'} \beta_\alpha A_\alpha = c, 
\ee
which concludes the proof. 
\end{proof}

Next, we show: 

 \begin{theorem}
    \label{thm:characterization_full_2}
    If, in addition to the hypothesis of Thm.~\ref{thm:characterization_full}, the gate $U$ also fulfills 
        \begin{align}
        (\mathbbm{1} \otimes \bra{-}) U = c \otimes \bra{-},
        \label{eq:BCS_app_def}
    \end{align}
 the matrices $\{u_\alpha\}_{\alpha=1}^r$ in the decomposition in Eq.~\eqref{eq:general_form_app} are bistochastic.
\end{theorem}

    \begin{proof}
        We act with $\mathbbm{1} \otimes \bra{-} $ on LHS and RHS of Eq.~\eqref{eq:BCS_app_def} and use Eq.~\eqref{eq:Usto_firstline} to obtain
        \begin{align}
        \text{RHS}& = (c \otimes \bra{-})= \left( \sum_{\alpha\in S} \lambda_\alpha \mu_\alpha A_\alpha \right)  \otimes \bra{-}\, ,\\
        \text{LHS}&= \sum_{\alpha\in S} \lambda_\alpha \mu_\alpha A_\alpha \otimes \bra{-} u_\alpha +\sum_{\alpha\in S'} \lambda_\alpha A_\alpha \otimes \bra{-} u'_\alpha .
        \end{align}
Since $\{A_\alpha\}$ are linearly independent we obtain  
\begin{align}
\bra{-} u_\alpha = \bra{-}, \qquad  \bra{-} u'_\alpha = 0\,.
\end{align}
Therefore, using the decomposition in Eq.~\eqref{eq:uuprime} we have that all $\{u_\alpha\}_{\alpha=1}^r$ appearing in Eq.~\eqref{eq:general_form_app} are bistochastic. This concludes the proof. 
\end{proof}

\section{Generalized condition}
\label{sec:appgencond}

The simplifications discussed in Sec.~\ref{sec:simplifications} can be repeated in the case of control gates fulfilling Eqs.~\eqref{eq:CS_app_def} and \eqref{eq:BCS_app_def} with $c$ bistochastic. Namely 
\begin{align}
\label{eq:CSAPP}
& U\ket{-}_2 = (c)_1 \ket{-}_2, & &\begin{tikzpicture}[baseline={([yshift=-0.6ex]current bounding box.center)},scale=0.65]
    \prop{0}{0}{colU}
    \MEh{0.5}{-0.5}
  \end{tikzpicture} = 
   \begin{tikzpicture}[baseline={([yshift=-0.6ex]current bounding box.center)},scale=0.65]
     \gridLine{0}{-0.5}{0}{0.5}
     \gridLine{0.5}{0}{0.5}{0.5}
     \sCircle{0}{.0}{yellow}
         \MEh{0.5}{0}
  \end{tikzpicture}\,, \\
  \label{eq:BCSAPP}
  &\bra{-}_2 U = (c)_1 \bra{-}_2,
   & &\begin{tikzpicture}[baseline={([yshift=-0.6ex]current bounding box.center)},scale=0.65]
    \prop{0}{0}{colU}
    \MEh{0.5}{0.5}
  \end{tikzpicture} = 
   \begin{tikzpicture}[baseline={([yshift=-0.6ex]current bounding box.center)},scale=0.65]
     \gridLine{0}{-0.5}{0}{0.5}
     \gridLine{0.5}{0}{0.5}{-0.5}
      \MEh{0.5}{0}
       \sCircle{0}{.0}{yellow}
         \MEh{0.5}{0}
  \end{tikzpicture},
\end{align}
where 
\begin{align}
& \bra{-} c = \bra{-}, 
 & &  \begin{tikzpicture}[baseline={([yshift=-0.6ex]current bounding box.center)},scale=0.65]
     \gridLine{0}{-0.5}{0}{0.5}
      \MEh{0}{.5}
       \sCircle{0}{.0}{yellow}
  \end{tikzpicture} =  \begin{tikzpicture}[baseline={([yshift=-0.6ex]current bounding box.center)},scale=0.65]
     \gridLine{0}{0}{0}{0.5}
      \MEh{0}{.5}
  \end{tikzpicture}, \\
  & c\ket{-}=\ket{-},  
 & &  \begin{tikzpicture}[baseline={([yshift=-0.6ex]current bounding box.center)},scale=0.65]
     \gridLine{0}{-0.5}{0}{0.5}
      \MEh{0}{-.5}
       \sCircle{0}{.0}{yellow}
  \end{tikzpicture} =  \begin{tikzpicture}[baseline={([yshift=-0.6ex]current bounding box.center)},scale=0.65]
     \gridLine{0}{-0.5}{0}{0}
      \MEh{0}{-.5}
  \end{tikzpicture}.
\end{align}
These conditions still lead to the same simplifications of the correlations. Sometimes,  we can use different pairing of single site gates in the brickwork, to reduce the gates to the form in the main text. One such example is $c_\alpha= a\ketbra{\alpha}b$, with $a,b$ single site bistochastic gates. In this case we can shift $a,b$ to $u_\alpha$ leading to a local gate $U'=\sum_\alpha \ketbra{\alpha} \otimes b u_\alpha a$.

\section{An example of random controlled gates leading to tilted East gate}
\label{sec:GenTiltedEast}
Consider two-qubit gates
\be
W = (w_1 \otimes w_2) ( \Id \otimes e^{i \beta X} ) e^{i J Z\otimes Z} (\Id \otimes e^{i \beta Y}) D (w_3\otimes w_4)
\ee
where $D$ a generic diagonal unitary matrix, while $\{w_j\}$ are $U(1)$ matrices of the form 
\be
w_j = e^{i \phi_j Z}\,. 
\ee
Note that the floquet quantum East gate introduced in~\cite{bertini2024localised} can be written in this form with $J=\pi/4$ (up to a left-right flip).

Averaging over random $\{w_j\}$ and following the procedure from Sec.~\ref{sec:Tostochastic} produces the gate 
\begin{equation*}
   U = \begin{tikzpicture}[baseline={([yshift=-0.6ex]current bounding box.center)},scale=0.65]
    \prop{0}{0}{colU}
  \end{tikzpicture} 
=  \Id  \otimes \ketbra{\mcirc} + (1-\cos(\beta)^2 (1-\sin(2J))) \ketbra{0}{0}\otimes \ketbra{\mcircf}+(1-\cos(\beta)^2 (1+\sin(2J))) \ketbra{1}{1}\otimes \ketbra{\mcircf}.
\label{eq:generalisation2}
\end{equation*}
Let us introduce $p=\sin(\beta)^2$ and $e^{-s}=\cos(2J)$:
\be
U=  \Id  \otimes \ketbra{\mcirc} + (1-p (1-e^{-s})) \ketbra{0}{0}\otimes \ketbra{\mcircf}+(1-p (1+e^{-s})) \ketbra{1}{1}\otimes \ketbra{\mcircf}.
\label{eq:generalisation4}
\ee
Note that if we write  Eq.~\eqref{eq:generalisation4} in the basis of $\ket{\mcirc},\ket{\mcircf}$ (perform Hadamard transformation on basis $\ket{0},\ket{1}$), we obtain the tilted stochastic East model studied in~\cite{defazio2024exact}:
\be
  U = \begin{pmatrix}
     1&0&0&0\\
     0&1&0&0\\
     0&0&1-p&p \ e^{-s}\\
     0&0&p \ e^{-s}&1-p
 \end{pmatrix}
\ee
For our physical values we can only have \emph{positive} $s$ tilting parameters in the parametrization of Ref.~\cite{defazio2024exact}, thus only probing the inactive phase.

\end{document}